\documentclass[12pt]{article}

\usepackage{complexity}
\usepackage{framed}
\usepackage{fullpage}
\usepackage{hyperref}

\usepackage{amsmath}
\usepackage{graphicx,graphics}
\usepackage{latexsym}

\usepackage[ruled]{algorithm2e}
\renewcommand{\algorithmcfname}{ALGORITHM}
\SetAlFnt{\small}
\SetAlCapFnt{\small}
\SetAlCapNameFnt{\small}
\SetAlCapHSkip{0pt}
\IncMargin{-\parindent}

\newlang{\GEN}{\sf GEN}
\newcommand{\T}[2]{\lang{T_{#2}^{#1}}}
\newcommand{\ft}[3]{\lang{FT_{#2}(#1, #3)}}
\newcommand{\bt}[3]{\lang{BT_{#2}(#1, #3)}}
\newcommand{\hft}[3]{\lang{FT^{#1}_{#2}(#3)}}
\newcommand{\sft}[3]{\lang{\widehat{FT}_{#2}(#1, #3)}}
\newcommand{\children}[3]{\lang{Children^{#1}_{#2}(#3)}}
\newcommand{\proj}[2]{\ensuremath{\pi(#1, #2)}}
\newcommand{\bit}[2]{\ensuremath{{(#1)_#2}}}
\DeclareMathOperator*{\cross}{\times}
\newcommand{\cart}[2]{\ensuremath{\cross\limits_{#1}^{#2}}}
\newcommand{\V}[2]{\ensuremath{v_{#1}^{#2}}}

\newcommand{\wpbl}[2]{\ensuremath{w(#2, #1)}}

\newcommand{\stheta}[1]{\ensuremath{\widetilde{\theta}\left({#1}\right)}}
\newcommand{\sbigoh}[1]{\ensuremath{\widetilde{O}\left({#1}\right)}}

\newtheorem{theorem}{Theorem}
\newtheorem{definition}[theorem]{Definition}
\newtheorem{proposition}[theorem]{Proposition}
\newtheorem{lemma}[theorem]{Lemma}

\newtheorem{remark}[theorem]{Remark}

\newtheorem{claim}[theorem]{Claim}
\newcommand{\qed}{\rule{5pt}{5pt}}
\newenvironment{proof}{\noindent{\bf Proof:}}{\qed\bigskip}

\begin{document}


\title{Pebbling, Entropy and Branching Program Size Lower Bounds
\footnote{A preliminary version of this work containing a subset 
of the results in this paper, appeared at the Symposium of Theoretical Aspects of Computer Science (STACS 2013).
The additional results in this extended version appears in the journal submission \cite{ks-journal} (June 2013).
}}
\author{Balagopal Komarath\footnote{Supported by the TCS Research Fellowship.}
~~and~~Jayalal Sarma M.N. \\
Department of Computer Science \& Engineering \\
IIT Madras, Chennai, India\\
Email : {\tt \{baluks,jayalal\}@cse.iitm.ac.in}}


\maketitle
\begin{abstract}
We contribute to the program of proving lower bounds on the size of
branching programs solving the Tree Evaluation Problem introduced by Cook et. al. in
~\cite{wehr-tep}. Proving a super-polynomial lower bound for the size of
nondeterministic thrifty branching programs would separate $\NL$
from $\P$ for thrifty models solving the tree evaluation problem. 
First, we show that {\em Read-Once Nondeterministic Thrifty BPs} are equivalent to whole 
black-white pebbling algorithms thus showing a tight lower bound (ignoring polynomial factors) 
for this model.

We then introduce a weaker restriction of nondeterministic thrifty 
branching programs called {\em Bitwise Independence}.
The best known \cite{wehr-tep} nondeterministic
thrifty branching programs (of size $O(k^{h/2+1})$)  for the tree evaluation problem are
Bitwise Independent.  
As our main result, we show that any Bitwise Independent
Nondeterministic Thrifty Branching Program solving \bt{h}{2}{k} must
have at least $\frac{1}{2}k^{h/2}$ states. Prior to this work, lower
bounds were known for nondeterministic thrifty branching programs only for fixed
heights $h=2,3,4$~\cite{wehr-tep}. We prove our results by associating a fractional black-white 
pebbling strategy with any bitwise independent nondeterministic
thrifty branching program solving the Tree Evaluation Problem. Such a
connection was not known previously even for fixed heights.

Our main technique is the entropy method introduced by Jukna and
Z{\'a}k \cite{JZ01} originally in the context of proving lower bounds for
read-once branching programs. We also show that the previous lower
bounds known \cite{wehr-tep} for deterministic branching programs for
Tree Evaluation Problem can be obtained using this approach. Using
this method, we also show tight lower bounds for any $k$-way
deterministic branching program solving Tree Evaluation Problem when
the instances are restricted to have the same group operation in all
internal nodes.
\end{abstract}






\section{Introduction} The question whether \L\ is a proper subset of \P\ is
one of the central problems in complexity theory. One of the approaches to the
problem was proposed as a program by Cook~\cite{Coo74} by introducing a
suitable computational problem, namely the solvable path systems. The program
suggests to prove lower bounds for increasingly stronger models of computation
solving the solvable path systems problem. Indeed, for the specific problem,
the attempt is to discover the structure of that restricted variant of the
underlying computation process that captures the most natural, and if possible
the most general, algorithmic strategies for solving the problem 
\cite{Coo74} \cite{Bar91} \cite{Edm99} \cite{Gal07} \cite{wehr-tep} \cite{wehr-gen}.
Cook~\cite{Coo74} also proved super-logarithmic space lower bounds for marking
machines solving the solvable path systems problem. Marking machines capture
pebbling algorithms (which is a class of ``natural'' algorithms) solving this
problem.

Barrington and Mckenzie \cite{Bar91} took a similar approach by considering the
problem \GEN\ and attempted to prove (upper and lower bounds) for increasingly
stronger models of computation for solving \GEN.  Specifically,
Barrington~\cite{Bar91} considered ``oracle'' branching programs where each
state of the branching program is allowed to ask a question about the input.
For example, a general BP can ask queries of the form ``What is the
$i^{\text{th}}$ bit of the input?'' (This is called a branching program with
BIT oracle).  Barrington~\cite{Bar91} proved exponential size lower bounds for
branching programs equipped only with certain ``weak'' oracles.  Gal {\em et
al} \cite{Gal07} considered incremental branching programs for solving \GEN\
which can be thought of as branching programs trying to solve the \GEN\ problem
by incrementally finding the elements of the closure.  

Cook {\em et al} \cite{wehr-tep} proposed the tree evaluation problem for
separating \L\ and \P\ and introduced thrifty branching programs as a model
that captures ``natural'' algorithms solving the tree evaluation problem.  It
is shown in \cite{wehr-tep} that deterministic thrifty branching programs
exactly correspond to algorithms that implement black pebbling.  They also
introduced the concept of fractional black-white pebbling and showed that
nondeterministic thrifty branching programs can implement fractional
black-white pebbling.  It is also known that super-polynomial size lower-bounds for
deterministic semantic incremental branching programs solving the \GEN\ problem
follows from super-polynomial size lower-bounds for deterministic thrifty branching
programs solving a generalization of tree evaluation problem called the
DAG-evaluation problem~\cite{wehr-gen}.

\noindent{\bf Tree Evaluation Problem:} We now briefly describe the
tree evaluation problem (see section~\ref{sec:prelims} for a formal
definition).  An instance of the tree evaluation problem, \ft{h}{2}{k}, is a
complete binary tree where each leaf is associated with an element from $[k]$
(which we think of as the value of the leaf node) and the $i^{\text{th}}$
internal node is associated with a function $f_i : [k]^{2} \mapsto [k]$. The
value of an internal node is obtained by applying this function to the values
of its children. The output is the value of the root node. The corresponding
Boolean version, \bt{h}{2}{k}, differs from \ft{h}{2}{k} in that the function
at the root node maps a value in $[k]^{2}$ to a value in $\{0, 1\}$. An
instance of \bt{h}{2}{k} is called a ``yes'' instance if and only if the value
of the root node is 1.  Another variant of the tree evaluation problem is the
single function variant \sft{h}{2}{k} where the functions at all internal nodes
are the same.  A natural computational model for tree evaluation problem is
$k$-way branching program where each node queries the value of a single $k$-ary
variable (i.e., the query is either $i$, where $i$ is a leaf node, or $f_j(r,
s)$, where $j$ is an internal node and $r, s \in [k]$). As observed
in~\cite{wehr-tep}, any size lower bound of the form $\Omega(k^{r(h)})$ for
$k$-way branching programs, where $r(h)$ is an unbounded function, would prove
that \L\ $\neq$ \P.  We only consider $k$-way branching programs in this paper.
Here we are interested in how the size of the branching programs solving
\ft{h}{2}{k} increases with respect to $h$ and $k$. 

A natural algorithm to solve \ft{h}{2}{k} is to evaluate the tree in a
bottom-up fashion. This can be captured by the concept of black pebbling
\T{h}{2}\ (The complete binary tree of height $h$). A black pebble on a node
indicates that the value of the node is known. A black pebble can be placed on
an internal node only if both its children are black pebbled. As a special
case, a black pebble can be freely placed on any leaf node.  It can be shown
that $h$ pebbles are necessary and sufficient for black pebbling \T{h}{2}.
Since a value in $[k]$ can be represented using $\log k$ bits. This corresponds
to a size bound of $\Theta(k^h)$ for branching programs. Similarly, fractional
black-white pebbling captures natural nondeterministic algorithms solving
\ft{h}{2}{k}. A white pebble can be freely placed on any node and corresponds
to guessing the value of that node. However, to remove a white pebble from a
node (this corresponds to verifying the guessed value) both its children have
to be pebbled. Moreover, a branching program may compute or guess a fraction of
bits of the values of nodes and this results in fractional black and white
pebbles respectively. 

A deterministic thrifty branching program is one in which the branching program
is only allowed to query $f_j(r, s)$ when $r$ and $s$ are the values of the
children of node $j$.  Cook {\em et al.}~\cite{wehr-tep} showed that
deterministic thrifty branching programs solving \bt{h}{2}{k} require
$\Omega(k^h)$ states by showing that such branching programs implement exactly
a black pebbling strategy. Cook {\em et al.}~\cite{wehr-tep} also proved tight
lower bounds for nondeterministic thrifty branching programs for $h = 2, 3,
4$. They also show an upper-bound of $O(k^{h/2+1})$ for nondeterministic
thrifty branching programs solving \ft{h}{2}{k}. This shows that the
nondeterministic variant is more powerful compared to the deterministic model. 

To complete the discussion, we refer the reader to \cite{Razborov91} for a detailed survey on branching program lower bounds. 
Specifically, we note that good lower bounds are known against read-once branching program models (See \cite{wegener87}, \cite{zak84}, \cite{JZ01}) although not for the problem that is of interest in this paper.

\noindent{\bf Our Results:} 
In this paper, we extend the results in Cook {\em et al.}~\cite{wehr-tep} to two restricted models in the nondeterministic setting. We also provide unified views of many results that were known regarding the branching program size for tree evaluation problem.

To begin with, we show that computation done by Read-Once nondeterministic thrifty branching programs can be captured by the whole black-white pebbling game. This observation combined with the known lower
bounds~\cite{wehr-tep} for whole black-white pebbling imply the following theorem.

\begin{theorem} Any Read-Once NTBP solving \bt{h}{2}{k} must have $k^{\lceil
h/2 \rceil}$ states.  \end{theorem}

As our main result, we show that computation of 
nondeterministic thrifty branching programs with a restriction
that we call \emph{bitwise independence} can be associated with a fractional
black-white pebbling sequence and therefore requires super-polynomial size. The
additional restriction of bitwise independence is not too severe since all
known upper-bounds using nondeterministic thrifty branching programs can be
achieved using those with bitwise independence property. In particular, the
branching program described in \cite{wehr-tep} that achieve $O(k^{h/2+1})$
upper-bound satisfy bitwise independence. Our main result is the first
super-polynomial lower bound for some restriction of nondeterministic thrifty
branching programs solving the tree evaluation problem.  \begin{theorem}[Main Theorem] If $B$ is a
bitwise independent nondeterministic thrifty branching program solving
\bt{h}{2}{k}, then $B$ has at least $\frac{1}{2}k^{h/2}$ states.  \end{theorem}

We associate these branching programs with fractional black-white pebbling. Cook {\em et
al.}~\cite{wehr-tep} showed that if the tree \T{h}{2}\ can be fractionally
pebbled using $p$ pebbles, then the corresponding (binary) Tree evaluation
problem can be solved by a nondeterministic thrifty branching program of size
$O(k^p)$. However, the converse direction is far from clear. We make progress
in this direction and prove our lower bound by connecting bitwise independent
nondeterministic thrifty branching programs to fractional black-white pebbling
sequences. We use the known result \cite{frac-pbl-lb} (see also~\cite{wehr-tep})
that $h/2+1$ pebbles are necessary and sufficient to pebble \T{h}{2} using
fractional black-white pebbling, to derive our lower bounds. We note that the
lower bounds for $h = 2, 3, 4$ in \cite{wehr-tep} were not shown by associating
it with fractional black-white pebbling. 

We summarise the relationships among pebbling games and branching program
models in Table~\ref{tab:bp-pbl}.  We use ``soft'' asymptotic notation that 
ignores factors polynomial in the input size (i.e., Factors of the order of 
$O((2^h k)^c)$ for any constant $c$) to describe the bounds in Table~\ref{tab:bp-pbl}.  
The exact bounds are given in the statements of appropriate theorems.

\begin{table}[ht] \begin{tabular}{llc} Thrifty Branching Program Model &
Pebbling Game & Size\\ \hline Deterministic & Black &
$\stheta{k^h}$ \cite{wehr-tep}\\ Nondeterministic Syntactic Read-Once & Whole Black-White
& $\stheta{k^{\lceil h/2 \rceil}}$ \\ Nondeterministic Bitwise Independent &
Fractional Black-White & $\stheta{k^{h/2}}$ \\ \end{tabular} \label{tab:bp-pbl}
\end{table}

Our main technique is a method proposed by Jukna and Z{\'a}k~\cite{JZ01} for
proving size lower bounds for branching programs which they call the {\em
entropy method}. Briefly, the method is to distribute a large set of inputs
among the states of the branching program such that only a small number of
inputs get mapped to any particular state. To achieve this, Jukna and
Z{\'a}k \cite{JZ01} proposed to consider the set $F$ of inputs reaching that state
and show that we can uniquely determine an input in $F$ by a decision tree of
low average depth (equivalently, the set $F$ has low entropy). It follows that
there are a large number of states.  

As our next contribution, we show that the lower bound proofs in \cite{wehr-tep} for $k$-way branching programs solving \hft{3}{2}{k},
\children{4}{2}{k} and thrifty branching programs solving \bt{h}{2}{k} can be
obtained using this framework.  Thus we get simplified and unified views of
the proofs for the following theorems.  \begin{theorem} \begin{itemize} \item
Any deterministic $k$-way branching program solving \hft{3}{2}{k} must have at
least $k^3$ states.  \item Any deterministic $k$-way branching program solving
\children{4}{2}{k} must have at least $k^4$ states.  \end{itemize}
\end{theorem}


We then apply our method in a restricted setting where the functions at all
internal nodes are given to be the same.  \begin{theorem} Any deterministic $k$
way branching program solving \sft{h}{2}{k} with the functions at internal
nodes restricted to a group operation must have at least  $2^{h-2} k$ states.
\end{theorem}

We observe that the above lower bound is tight. Indeed, when the internal
operation is that of a group, the associativity property can be used to design
branching programs of size $O(2^h k)$, when the function at the internal nodes
is fixed.  When the function at the internal nodes is also a part of the input,
the upper bound is off by a factor of $k$, namely $O(2^h k^2)$.

The rest of the paper is organized as follows: In Section~\ref{sec:prelims} we
introduce the preliminaries needed for the paper. We prove the main result in
section~\ref{sec:bintbp}. We consider read-once nondeterministic thrifty
BPs in section~\ref{sec:rontbp}.  Further applications of the entropy method
are described in Section~\ref{sec:otherproofs}.

\section{Preliminaries} \label{sec:prelims}
For definitions of basic notions in complexity theory, we refer the reader to a
standard textbook~\cite{arora-barak,vollmertext}. 
Let $[k] = \{ 0, \ldots , k-1 \}$.  We give the formal definition
of tree evaluation problems first. We will use the term \emph{node} to refer to the 
vertices in the tree referred to by the 
input instance and the term \emph{state} to refer to vertices in the 
branching program.  In the following discussion, we label the
nodes of the tree using usual heap numbering.  The root node is labelled $1$
and for each internal node $i$, its left child is labelled $2i$ and right child
is labelled $2i + 1$.  We use $v_i$ to denote the value of the $i^{\text{th}}$
node in the input tree. When we are talking about a specific input $I$, we use
$v_{i}^{I}$ to denote the value of node $i$ of the input $I$.

We now define the function and boolean versions of the tree evaluation problem. 

\begin{definition}{(Tree Evaluation Problems \cite{wehr-tep})}
Input: The tree \T{h}{2} where each leaf node is associated with a value from
$[k]$ and each internal node $i$ is associated with a function $f_i : [k]^{2}
\mapsto [k]$, where $h, k \geq 2$

Output for \ft{h}{2}{k}: The value $v_1 \in [k]$ of the root node, where in
general $v_i = a$ if $i$ is a leaf and $a$ is the value associated with
$i^{\text{th}}$ node in the input and $v_i = f_i(v_{2i}, v_{2i+1})$ if $i$ is a
non-leaf node.

Output for \bt{h}{2}{k}: The value $v_1 \in \{0,1\}$ of the root node. The
evaluation rules are the same as for \ft{h}{2}{k}.
\end{definition} 

It is known that tree evaluation problems are in \LOGDCFL\ \cite{wehr-tep} (For
definition of \LOGDCFL\ see \cite{sud-logcfl}). Note that the input size when
represented in binary is $O(2^h k^2 \log k)$.  Since all values in the
definition of tree evaluation problems are $k$-ary, a general model to solve
tree evaluation problems is a branching program that queries $k$-ary variables.
Such branching programs are called \emph{$k$-way branching programs} (or BP in short in this paper), 
since each query has $k$ possible outcomes (depending on the value of the queried variable.). We
define these models formally now.  

\begin{definition}[$k$-way Branching Program (BP) \cite{wehr-tep}]
A nondeterministic $k$-way branching program $B$ for \ft{h}{2}{k} is a directed 
multi-graph.  It consists of a designated start state and $k$ final states labelled $0, \ldots , k-1$. 
A non-final state is labelled either $\ell_i$ where $i$ is a leaf node or labelled 
$f_i(x_1, x_2)$ where $i$ is an internal node, $x_1, x_2 \in [k]$, 
and each outgoing edge is labelled by an element from $[k]$. 
A computation path on input $I$ is a directed path from the start state and each 
edge in the path is consistent with $I$. At least one
such computation must end in the final state labelled \V{1}{I} and all computations ending in a final state 
must end in the final state labelled \V{1}{I}. The BP $B$ is deterministic if and
only if each query state has exactly $k$ outgoing edges labelled $0 , \ldots ,k-1$. 

A nondeterministic $k$-way branching program $B$ for \bt{h}{2}{k} is defined similarly except
that each query state labelled $f_1(x_1, x_2)$ where $x_1, x_2 \in [k]$ has all
of its outgoing edges labelled by either $0$ or $1$. There is a designated accepting state that has no outgoing edges. 
The BP $B$ solves \bt{h}{2}{k} if and only if for every ``yes''
instance $I$ it has at least one computation path from the start state to the accepting state consistent with $I$ 
(an accepting computation path)
 and for every ``no'' instance the BP has no accepting computation path.
The BP $B$ is deterministic if and only if each query
state labelled $f_1(x_1, x_2)$ has exactly two outgoing edges labelled $0$ and
$1$ and every other query state has exactly $k$ outgoing edges labelled $0,
\ldots , k-1$. \end{definition}

By a sub-BP $B^{\prime}$ of $B$ obtained by restricting input set $E$ to
$E^{\prime}$, we refer to the BP obtained from $B$ by removing edges not used
by inputs in $E^{\prime}$ and by shortcutting states for which only one
outgoing edge can be active when we consider computation on instances in
$E^{\prime}$.

The size of binary branching programs solving tree evaluation problems differ 
from the size of $k$-way branching programs by a factor of at most $k$. 
Therefore, a size lower bound of $\Omega(k^{r(h)})$ for $k$-way branching programs, where $r(h)$ is an
unbounded function, would separate \L\ from \LOGDCFL.

\begin{definition}[Nondeterministic Thrifty BP (NTBP) \cite{wehr-tep}]
A nondeterministic BP solving \bt{h}{2}{k} is called \emph{thrifty} if and
only if for any accepting computation path on any instance $I$ any query state
labelled $f_i(x_{1}, x_{2})$ satisfies $x_{1} = v_{2i}$ and $x_2 = v_{2i+1}$
(i.e., the internal nodes are queried only at the correct values of its
children).
\end{definition}

By a state querying node $i$ we mean that the state queries $f_i(x, y)$ for some 
$x$, $y \in [k]$ when $i$ is an internal node and that the state queries $\ell_i$ when 
$i$ is a leaf node.

\begin{definition}[Syntactic Read-Once NTBP (RONTBP)]
An NTBP solving \bt{h}{2}{k} is called \emph{syntactic read-once} if and only if 
any graph-theoretic path from the start state to the accept state queries each node 
at most once. 
\end{definition}

Let $N = 2^h - 2$ be the total number of non-root nodes in \T{h}{2}.  Let $B$
be a nondeterministic thrifty BP for \bt{h}{2}{k}.  Let $s$ be a state of $B$.
We define 

\begin{align*} F_s &= \{ (\V{2}{I}, \ldots , \V{N+1}{I}) \colon
\exists I \text{ and a computation path $C(I)$} \text{ such that } s \in C(I)
\}\\ A_s &= \{ (\V{2}{I}, \ldots , \V{N+1}{I})  \colon \exists I \text{ and an
accepting computation path $C(I)$} \text{ such that } s \in C(I) \}
\end{align*}

We use $\proj{S}{i}$ to denote the set of all $i^{\text{th}}$ component of the
tuples in $S$ (typically, $S$ is either $F_s$ or $A_s$ for some $s$.). That is,
the set formed by projecting the $i^{\text{th}}$ component of all tuples in
$S$.  For any encoding function $\phi : [k] \mapsto \{ 0, 1 \}^{\lceil \log k
\rceil}$, we use \bit{r}{i} to denote the $i^{\text{th}}$ bit of $r \in [k]$
when $r$ is encoded using $\phi$.

\begin{definition}[Bitwise-independent NTBP (BINTBP)]
Let $k = 2^{\ell}$ and let $B$ be a nondeterministic thrifty BP solving
\bt{h}{2}{k}. Then $B$ is called \emph{bitwise independent} if and only if there exists an
encoding function $\phi : [k] \mapsto \{0, 1\}^{\ell}$ such that for every
state $s$ in $B$ the following two conditions are satisfied.

\begin{align*} F_s &= \cart{i = 2}{N+1} \phi^{-1}\left(\cart{j = 1}{\ell}
\bit{\proj{F_s}{i}}{j}\right) \\A_s &= \cart{i = 2}{N+1} \phi^{-1}\left(\cart{j
= 1}{\ell} \bit{\proj{A_s}{i}}{j}\right) \end{align*}

Here the outer Cartesian product is the normal Cartesian product and the
inner one concatenates all the bits after forming the
Cartesian product.  When $k$ is not a power of 2, we consider the largest
power of 2 smaller than $k$. Let this be $2^{\ell}$.  Then $B$ is
bitwise independent if and only if the sub-BP $B^{\prime}$ of $B$ obtained by considering only inputs 
where all values are from $[2^\ell]$ is bitwise independent. 
\label{def:bintbp}
\end{definition}

The intuition is that at any state in the BP the bits of values of non-root nodes can be
partitioned into ``fixed'' bits and ``unfixed'' bits and the sets $F_s$ and
$A_s$ are such that all possible combinations of unfixed bits are in the set.
i.e., the BP cannot store implicit information about bits (such as, the second
bit is the complement of the first bit). 

If we only consider minimal bitwise independent nondeterministic thrifty BPs,
then we have $|F_s|, |A_s| \geq 1$ for any query state $s$. This is because any
query state $s$ that does not have any accepting computation path passing
through it can be removed. Also note that by the
definition of bitwise independence, for any $i$ and $s$, we have
$\proj{F_s}{i}$ and $\proj{A_s}{i}$ are always powers of two when $k$ is a
power of 2.

The following input set will be used to prove lower bounds for thrifty BPs.  We note that 
this set was also used in \cite{wehr-tep} to prove lower bounds for deterministic thrifty BPs 
solving \bt{h}{2}{k}.

\begin{definition}[Hard Inputs for Thrifty BP]
\begin{align*}
E = \{ I : & f_1^{I}(x, y) = 1 \text{ for all } x, y \in [k] \\
	& f_i^{I}(x, y) \in [k] \text{ if } x = \V{2i}{I} \text{ and } y = \V{2i+1}{I} \text{ for all internal nodes } i\\
	& f_i^{I}(x, y) = 0 \text{ if } x \neq \V{2i}{I} \text{ or } y \neq \V{2i+1}{I} \text{ for all internal nodes } i\\
	& \ell_i^{I} \in [k] \text{ for all leaf nodes } i\}
\end{align*}
\label{def:hard-inputs}
\end{definition}

Here we set $f_1$ to the constant function 1 and we allow all $k$-ary values to take arbitrary values if they can be queried by a 
thrifty BP and set them to 0 otherwise.  Note that all inputs in $E$ are ``yes'' instances and $|E| = k^N$.

We make the following observation about accepting computation paths for inputs in $E$ in any NTBP solving \bt{h}{2}{k}.

\begin{proposition}
Let $B$ be an NTBP solving \bt{h}{2}{k}.  Let $I \in E$ and let $C(I)$ be an accepting computation path for $I$ in $B$, then 
all nodes are queried in $C(I)$. 
\label{prop:thrifty-all-nodes-queried}
\end{proposition}
\begin{proof}
If the root node is not queried for some $I \in E$, then the input $I^{\prime}$ which is the same as $I$ but with $f_1 = 0$ 
is also accepted by $B$.  Let $i$ be some non-root node and assume that $C(I)$ does not have a state querying node $i$.  Then 
the input $I^{\prime}$ which is the same as $I$ but with $f_{i}^{I^{\prime}}(\V{2i}{I}, \V{2i+1}{I}) \neq f_{i}^{I}(\V{2i}{I}, \V{2i+1}{I})$ 
is also accepted by $C(I)$.  But then $C(I)$ makes a non-thrifty query when querying the parent node of $i$ for either $I$ or $I^{\prime}$.  
Therefore $C(I)$ does not query the parent of $I$.  By induction, we can conclude that $C(I)$ does not query the root node which is a 
contradiction.
\end{proof}

\subsection{Pebbling}

Pebbling sequences capture ``natural'' algorithms solving 
tree evaluation problems by computing the values of nodes 
of the tree in a bottom-up fashion. 
\begin{definition}[Fractional Black-White, Whole Black-White and Black Pebbling \cite{wehr-tep}]
\label{def:pebbling}
A fractional black-white pebbling configuration of a rooted binary tree \T{h}{2} is an
assignment of a pair of real numbers $(b(i), w(i))$ to each node $i$ of the
tree. The values $b(i)$ and $w(i)$ are called the black and white pebble
values, respectively, of node $i$. We have for every $i$ \begin{align}
\label{eq:pebble-restr} b(i) + w(i) &\leq 1 \notag \\ 0 \leq b(i), w(i) &\leq 1
\end{align}

The legal pebble moves are as follows.

\begin{enumerate}
\item For any node $i$, decrease $b(i)$ arbitrarily.  \item For any node $i$,
increase $w(i)$ arbitrarily.  \item For any node $i$, if each child of $i$ has
pebble value $1$, then decrease $w(i)$ to 0.  \item For any node $i$, if each
child of $i$ has pebble value $1$, then increase $b(i)$ arbitrarily and 
simultaneously decrease the black pebble value of a child of $i$.
\end{enumerate}

The number of pebbles in a configuration is the sum over all nodes $i$ of $b(i)
+ w(i)$. A fractional black-white pebbling of \T{h}{2} using $p$ pebbles is a sequence of
(legal) fractional black-white pebbling moves on nodes of \T{h}{2} which starts and ends
with each node having pebble value $0$ and at some point the root node has
black pebble value $1$, and no configuration has more than $p$ pebbles.

A whole black-white pebbling is a fractional black-white pebbling such that for all configurations and all nodes $i$, 
we have $b(i), w(i) \in \{0,1\}$.

A black pebbling is a whole black-white pebbling such that for all configurations and for all nodes $i$, we have $w(i) = 0$.
\end{definition}

We now give an intuitive description of Definition~\ref{def:pebbling}.  
A black pebble value of $\epsilon$ at a node indicates that $\epsilon \log k$ bits of 
the value of that node is known to the BP. Similarly, a white pebble value of $\epsilon$ 
indicates that $\epsilon \log k$ bits of the value of that node has been guessed by the BP 
(equivalently, the BP needs to verify $\epsilon \log k$ bits of the value of that node). 
The pebbling rules capture the intuition that in order to compute or verify
(a fraction of) the value of any node, the BP must completely figure out (by
computing or guessing) the values of its children. 

It is known that $h$, $\lceil h/2 \rceil + 1$ and $h/2 + 1$ pebbles are
required for black pebbling, whole black-white pebbling and fractional
black-white pebbling (resp.) \T{h}{2} (See \cite{wehr-tep}, \cite{frac-pbl-lb}).

\begin{definition}[Read-Once Whole Black-White Pebbling]
A whole black-white pebbling $C_1, \ldots, C_t$ of \T{h}{2} is called \emph{read-once} if and only if for any node $n$ 
there exists $i$ and $j$, with $i < j$, such that
\begin{itemize}
\item For $k < i$, we have $b(n) = w(n) = 0$ for $C_k$.
\item The pebble values $b(n)$ and $w(n)$ remains same from $C_i$ through $C_j$ and either $b(n) \neq 0$ or $w(n) \neq 0$
\item For $k > j$, we have $b(n) = w(n) = 0$ for $C_k$.
\end{itemize}
\end{definition}
\subsection{Entropy Method} \label{subsec:entropy} We now formally describe the
entropy method introduced by Jukna and Z{\'a}k in \cite{JZ01}. We specialize the description slightly
to suit our application of the method. Let $B$ be a BP computing the
characteristic function of language $\lang{L}$.  Let $A$ be a particular set of
inputs and let $States(B)$ denote the set of all non-final states of the BP $B$. 
Define a ``distribution'' function $g: A \mapsto States(B)$.  Now
consider an arbitrary state $s$ in the range of $g$ and let $F = g^{-1}(s)$.
Define a decision tree $D$ such that each $a \in F$ reaches a unique leaf in
$D$.  Such a decision tree is called a `splitting tree' for $F$ in \cite{JZ01}.
The next step is to prove that $D$ has low depth which will imply that the
entropy of $F$, $h(F) = \log |F|$, is small.  Then we have $Size(B) \geq 2^{|A|
- h(F)}$. In defining $A$ and $g$, we may use properties of $\lang{L}$ and any
  restrictions imposed on the structure of $B$. The goal is to minimize the
maximum value of $h(F)$ over all choices of $F$ and at the time using an $A$ that is large enough to 
get the required lower bounds.

\section{Tight Bounds for RONTBP} \label{sec:rontbp}
We prove upper bounds for RONTBP by showing that they can implement read-once 
whole black-white pebbling to solve \bt{h}{2}{k}.

\begin{proposition} \label{prop:ropbl} There is a read-once whole black-white
pebbling of \T{h}{2} using $\lceil h/2 \rceil + 1$ pebbles.
\end{proposition} \begin{proof} Cook et. al. \cite{wehr-tep} has given a whole
black-white pebbling of \T{h}{2} using ${\lceil h/2 \rceil + 1}$ pebbles.  We
use $T_i$ to denote the subtree rooted at node $i$. The pebbling strategy in 
\cite{wehr-tep} is given below for completeness.  We describe the pebbling 
procedure for height $h+2$ tree assuming height $h$ tree has a whole black-white pebbling procedure.  
The induction hypothesis is that \T{h}{2} can be pebbled using $N(h) = \lceil h/2 \rceil + 1$ pebbles and 
there is a critical time in the pebbling of \T{h}{2} such that the root node has a black 
pebble and the tree has at most $N(h) - 1$ pebbles.  This is true for \T{2}{2} because we 
can place two black pebbles on leaves and then slide one to the root and remove the other.  
Now the root has a black pebble and the tree has $N(h) - 1 = 1$ pebble on it.

\begin{enumerate}
\item Place a black pebble on node $4$ by running the pebbling procedure 
on $T_4$.
\item Run the pebbling procedure on $T_5$, Stop when node $5$ has a black 
pebble on it.
\item Slide the black pebble on node $4$ to node $2$.
\item Remove black pebble on node $5$.
\item Resume the pebbling for $T_5$ and run it to completion.
\item Run the pebbling procedure on $T_6$ and suspend when node $6$ has 
a black pebble.
\item Place a white pebble on node $7$.
\item Slide the black pebble on node $6$ to node $3$.
\item Slide the black pebble on node $2$ to root node.
\item Remove black pebble from node $3$. (This is the critical time for \T{h+2}{2})
\item Remove black pebble from root node.
\item Resume the pebbling for $T_6$ and run it to completion.
\item Remove the white pebble on node $7$ by running the pebbling procedure 
for $T_7$.
\end{enumerate}

It is easy to see that this pebbling strategy is read-once.  In particular, we
stress that the pebbling strategy only suspends the pebbling of subtrees and
does not remove any pebbles from it until the pebbling for those subtrees are
resumed (This is being done in Steps~(2) and (5) and Steps~(6) and (12)).  
Since \T{2}{2} can be pebbled using $\lceil 2/2 \rceil + 1 = 2$
pebbles in a read-once fashion, it follows by induction that the above pebbling strategy for \T{h}{2}
is read-once.  \end{proof}

\begin{theorem}
\label{thm:rontbp-upper-bound}
There is a RONTBP solving \bt{h}{2}{k} using at most $(2^h - 1) k^{\lceil h/2 \rceil + 1}$ states.
\end{theorem}
\begin{proof} 
The construction uses the same idea used by Cook et. al. in \cite{wehr-tep} 
to construct an NTBP solving \bt{h}{2}{k}.  
Let $C_1, \ldots, C_t$
be the read-once whole black-white pebbling of \T{h}{2} given by
Proposition~\ref{prop:ropbl}.  We now describe a RONTBP $B$ that uses this
read-once pebbling to solve \bt{h}{2}{k}.  The RONTBP $B$ has $t$ layers
numbered $1$ to $t$.  The first layer consists of only the start state and
the last layer consists of only the accepting state.  Let $B_i$ denote the set
of all nodes with a black pebble on them in pebbling configuration $C_i$ and
let $W_i$ denote the set of all nodes with a white pebble on them in $C_i$.
The $i^{\text{th}}$ layer of $B$ has $k^{|B_i| + |W_i|}$ states.  We ``tag''
each state in layer $i$ with a set of possible values for these pebbled nodes
such that we will have exactly one state for each setting of possible values
for these pebbled nodes.  We stress that this ``tag'' is only used to make the
description of the RONTBP easier.  For a state $s$, we denote its tag by
$tag(s)$.  We can now desribe the labelling of states of $B$ and the edges of
$B$ (which will correspond to the pebbling moves) easily using these tags.  We
describe the edges from layer $i$ to layer $i+1$ in terms of the pebbling move
that takes the pebbling configuration $C_i$ to $C_{i+1}$.  

\begin{description}
\item [Place a black pebble on $j$] All states in layer $i$
are labelled $\ell_j$ if $j$ is a leaf node.  Otherwise each state $s$ in layer $i$ is labelled 
$f_j(x, y)$ where $x$ and $y$ are the values of $v_{2j}$ and $v_{2j+1}$
respectievely in $tag(s)$.  We direct the outgoing edge labelled $v$ from $s$
to the state $s^{\prime}$ in layer $i+1$ such that $tag(s^{\prime}) = tag(s)
\cup \{ v_j = v \}$ for each $v \in [k]$.

\item [Place a white pebble on $j$] All states in layer $i$ are unlabelled (they 
are ``guess'' states) and all edges from layer $i$ to layer $i+1$ are unlabelled.  
From each state $s$ in layer $i$, add an unlabelled outgoing edge from $s$ 
to each $s^{\prime}$ in layer $i+1$ that satisfies $tag(s^{\prime}) = tag(s) \cup \{ v_j = v \}$ for some 
$v \in [k]$.

\item [Remove a black pebble from $j$] All states in layer $i$ are unlabelled (they are 
``forget'' states).  From each state $s$ in layer $i$ such that $tag(s)$ contains $v_j = v$ for some 
$v \in [k]$, we add an unlabelled edge to each $s^{\prime}$ in layer $i+1$ 
that satisfies $tag(s^{\prime}) = tag(s) - \{ v_j = v \}$.

\item [Remove a white pebble from $j$] For each state $s$ where $tag(s)$ contains 
$v_j = v$ for some $v \in [k]$, we label $s$ with $\ell_j$ if $j$ is a leaf node. 
Otherwise $j$ is an internal node and we label $s$ with $f_j(x, y)$ where $x$ and $y$ are the 
values of $2j$ and $2j+1$ in $tag(s)$.  We add a single outgoing edge from $s$ labelled 
$v$ to the state $s^{\prime}$ in layer $i+1$ such that 
$tag(s^{\prime}) = tag(s) - \{ v_j = v \}$.

\item [Slide a black pebble from $j$ to is parent $j^{\prime} = \lfloor j/2 \rfloor$]
Each state $s$ in layer $i$ is labelled $f_{j^{\prime}}(x, y)$ where $x$ and $y$ are the 
values of $j$ and its sibling in $tag(s)$.  Add $k$ outgoing edges from $s$ labelled $0$ 
to $k-1$ such that the edge labelled $v$ is directed to $s^{\prime}$ in layer $i+1$ such that 
$tag(s^{\prime}) = tag(s) \cup \{ v_{j^{\prime}} = v \} - \{ v_j = x \}$ (assuming value of $j$ is 
$x$). 

\item [Slide a black pebble to the root node]
Each state $s$ in layer $i$ is labelled $f_{1}(x, y)$ where $x$ and $y$ are the 
values of nodes $2$ and $3$ in $tag(s)$.  Add an outgoing edge labelled $1$ to $s^{\prime}$ in 
layer $i+1$ such that $tag(s^{\prime}) = tag(s) - \{ v_2 = x \}$ (assuming that the black pebble was 
slid from node $2$ to the root). 
\end{description}

We can remove the unlabelled states with unlabelled edges by the following procedure.  If there is an 
edge labelled $v$ from some state $s$ to an unlabelled state $s^{\prime}$ with $e$ outgoing edges, 
then remove the state $s^{\prime}$ and add $e$ outgoing edges labelled $v$ from $s$ to the out-neighbors of 
$s^{\prime}$.  It is easy to see that this BP computes the same function as the original one.  The number of 
non-final states of $B$ is at most $(2^h - 1) k^{\lceil h/2 \rceil + 1}$ as there are $(2^h - 1)$ layers and any layer 
contains at most $k^{\lceil h/2 \rceil + 1}$ states.  Finally, the BP $B$ is a RONTBP since it implements 
a read-once pebbling of \T{h}{2}.
\end{proof}

We now prove tight lower bounds for size of RONTBPs solving \bt{h}{2}{k}.  The idea is to associate the computation 
of a RONTBP with a whole black-white pebbling.  We associate a whole black-white pebbling configuration with each 
state in the RONTBP such that if we take an accepting computation path of any instance in $E$, the sequence of 
pebbling configurations along the computation path is a valid pebbling of \T{h}{2}.  
Then we proceed to show that if we consider a state 
$s$ that has at least $\lceil h/2 \rceil + 1$ pebbles on a computation path (such a state exists on any 
accepting computation path), then the number of inputs reaching $s$ must be small.  
In particular, for any input $I$ on an accepting computation path reaching $s$, the 
values of pebbled nodes can be inferred from the state $s$ and the values of unpebbled nodes.  
This shows that the state $s$ must encode an element 
from a set of $k^{p}$ values where $p$ is the number of pebbled nodes.  The lower bound follows.

The following definition tells us how to extract a whole black-white pebbling from a RONTBP.

\begin{definition}[Pebbling Configuration at a State]
\label{def:rontbp-pbl-def}
Let $B$ be a RONTBP solving \bt{h}{2}{k}.  
Let $I \in E$ be arbitrary and let $C(I)$ be an arbitrary accepting computation path for $I$ in $B$.  Then the pebble 
value of a non-root node $i$ with parent $i^{\prime} = \lfloor i/2 \rfloor$ in the pebbling configuration 
associated with $s$ is defined as 
\begin{itemize}
\item If the state querying node $i^{\prime}$ comes after $s$ (or $i^{\prime}$ is queried by $s$) and the state querying node $i$ 
comes before $s$ in $C(I)$, then the node $i$ is black pebbled at state $s$.
\item If the state querying node $i$ comes after $s$ (or $i$ is queried by $s$) and the state querying node $i^{\prime}$ 
comes before $s$ in $C(I)$, then the node $i$ is white pebbled at state $s$.
\item Otherwise, the node $i$ is unpebbled at state $s$.
\end{itemize}
\end{definition}

For deriving a pebbling of \T{h}{2} from $C(I)$, we can assume that the root node is pebbled and 
unpebbled at the state immediately following the state querying the root node.  This does not affect 
the lower bound since the value at the root node is always 1 for any input in $E$.

Now we show that the pebbling configuration at some state $s$  defined above is independent of the input 
$I$ and the accepting computation path $C(I)$ that passes through state $s$.  In other words, this shows that 
the pebbling configuration at a state only depends upon the state $s$.  Note that if there are no accepting computation 
paths passing through a state $s$ in a RONTBP, then that state can be deleted from the RONTBP.

\begin{lemma}
\label{lem:pbl-indep}
Let $B$ be a RONTBP solving \bt{h}{2}{k}. Then the pebbling configuration at any state $s$ in $B$ depends only on $s$.  In particular, it is independent of any 
input $I$ and any accepting computation path $C(I)$ used to define it. 
\end{lemma}
\begin{proof}
Let $I$ and $I^{\prime}$ be two inputs in $E$ with accepting computation paths $C$ and $C^{\prime}$ passing through the state $s$.  Consider an arbitrary node 
$i$ and we argue that the pebble value of node $i$ with respect to $C$ is the same as pebble value of node $i$ with respect to $C^{\prime}$.  Let 
$i^{\prime} = \lfloor i/2 \rfloor$ be the parent of $i$.  We consider three cases based on the pebble value of node $i$ at $s$ with respect to $C^{\prime}$.  

\begin{description}
\item [Node $i$ is black pebbled] By Definition~\ref{def:rontbp-pbl-def}, we have a state $r$ querying $i$ before $s$ and a state $t$ querying $i^{\prime}$ after 
$s$ in the computation path $C$ (It is possible that $t = s$).  
Now the state $r^{\prime}$ querying $i$ on $C^{\prime}$ must precede state $s$ in the computation path $C^{\prime}$.  Otherwise the path 
$start \leadsto r \leadsto s \leadsto r^{\prime} \leadsto accept$ is a path in $B$ that queries node $i$ twice which is not possible since $B$ is a RONTBP.  
Similarly, the state $t^{\prime}$ querying $i^{\prime}$ must come after $s$ on $C^{\prime}$ (It is possible that $t = t^{\prime} = s$).

\item [Node $i$ is white pebbled] By Definition~\ref{def:rontbp-pbl-def}, we have a state $r$ querying $i^{\prime}$ before $s$ and a state $t$ querying $i$ after 
$s$ in the computation path $C$ (It is possible that $t = s$).  Now the state $r^{\prime}$ querying $i$ on $C^{\prime}$ must come before 
state $s$ in the computation path $C^{\prime}$.  Otherwise the path $start \leadsto r \leadsto s \leadsto r^{\prime} \leadsto accept$ is 
a path in $B$ that queries node $i$ twice which is not possible since $B$ is a RONTBP.  
Similarly, the state $t^{\prime}$ querying $i^{\prime}$ must come before $s$ on $C^{\prime}$ (It is possible that $t = t^{\prime} = s$).

\item [Node $i$ is not pebbled] We have two cases to consider.
\begin{description}
\item [Nodes $i$ and $i^{\prime}$ are queried before $s$ (With one of them possibly queried at $s$)] On $C^{\prime}$ both $i$ and $i^{\prime}$ must be 
queried before (or at) $s$ as otherwise we can construct a path from start state to accepting state that queries some node at least twice.
\item [Nodes $i$ and $i^{\prime}$ are queried before $s$ (With one of them possibly queried at $s$)] On $C^{\prime}$ both $i$ and $i^{\prime}$ must be 
queried after (or at) $s$ as otherwise we can construct a path from start state to accepting state that queries some node at least twice.

\end{description} 
\end{description}
The lemma follows.
\end{proof}

We now describe an algorithm FindPebbled that outputs a list of candidate values for the non-pebbled nodes at a state.  We 
describe the algorithm using nondeterminism.  In each nondeterministic path the algorithm may or may not output candidate 
values for the non-pebbled nodes.  The list output by the algorithm consists of all outputs taken over all nondeterministic paths.  
The nondeterminism can be easily eliminated using standard techniques.  We stress that the efficiency of the algorithm is of no concern 
as the algorithm is only a tool for proving the lower bound for RONTBPs.

\begin{algorithm}
\SetKwData{CurState}{curstate}
\SetKwData{Nil}{nil}
\SetKwInOut{Input}{input}\SetKwInOut{Output}{output}

\Input{$s$: A state in the RONTBP $B$,  $I_n$: The values of nodes that are not pebbled at $s$ for an input $I \in E$ that has an accepting computation path passing through $s$}
\Output{$LI_p$: A list of possible values for nodes that are not pebbled at $s$ such that there is a consistent path (with $I_n | I_p$) from $s$ to the accepting state}

\CurState$\leftarrow s$\;
\While{\CurState is not the accept state or \Nil}{
Let \CurState query $f_{i}(x, y)$\;
\tcc{Note that nodes $2i$ and $2i+1$ cannot be white pebbled at $s$ as it would mean that the node $i$ was queried before $s$ in $B$}
\If{node $2i$ is black pebbled at $s$}{
Set the value of node $2i$ to $x$
}
\ElseIf{node $2i$ is not pebbled at $s$ and the value of node $2i$ in $I_n$ is not $x$}{
\CurState $\leftarrow$ \Nil\;
Exit while loop\;
}
\If{node $2i+1$ is black pebbled at $s$}{
Set the value of node $2i+1$ to $y$
}
\ElseIf{node $2i+1$ is not pebbled at $s$ and the value of node $2i+1$ in $I_n$ is not $y$}{
\CurState $\leftarrow$ \Nil\;
Exit while loop\;
}
\tcc{Note that the node $i$ cannot be black pebbled at $s$ as it would mean that the node $i$ was queried before $s$ in $B$}
\If{$i$ is white pebbled at $s$}{
Nondeterministically follow one of the outgoing edges from \CurState after setting the value of node $i$ to the label of the outgoing 
edge chosen and set \CurState to the new state
}
\ElseIf{$i$ is not pebbled at $s$}{
Nondeterministically follow one of the outgoing edges consistent with the value of node $i$ in $I_n$ 
and set \CurState to the new state
}
}
\If{\CurState is the accept state}{
Output the values for nodes that are not pebbled found in this nondeterministic path
}
\caption{FindPebbled}
\end{algorithm}

\begin{algorithm}
\SetKwInOut{Input}{input}\SetKwInOut{Output}{output}
\SetKwFunction{FindPebbled}{FindPebbled}

\Input{$s$: A state in the RONTBP $B$,  $I_n$: The values of nodes that are not pebbled at $s$ for an input $I \in E$ that has an accepting computation path passing through $s$}
\Output{$I$: The unique input in $E$ that is consistent with $I_n$ and has an accepting computation path through $s$}

$LI_p \leftarrow$ \FindPebbled{$s$, $I_n$}\;
\For{each $I_p$ in $LI_p$}{
Simulate $B$ with $I = I_n | I_p$ and if there is an accepting computation path for $I$ through $s$ output $I$
}
\caption{FindInput}
\label{alg:FindInput}
\end{algorithm}
\begin{lemma}
\label{lem:state-input}
Let $B$ be a RONTBP solving \bt{h}{2}{k}.  Let $I \in E$ and let $C(I)$ be an accepting computation path 
for $I$ that passes through a state $s \in States(B)$.  Let $p$ be the number of pebbled non-root nodes in the 
pebbling configuration associated with $s$.  Then Algorithm~\ref{alg:FindInput} outputs $I$ when given 
as input the state $s$ and $N - p$ values in $[k]$ which correspond to values of nodes in $I$ 
that do not have any pebble in the configuration associated with $s$. 
\end{lemma}
\begin{proof}
Let $I_n$ denote the values of nodes that are not pebbled at state $s$ given as input to FindInput.  
Suppose for contradiction that FindInput outputs $I = I_n | I_p$ and $I^{\prime} = I_n | I_{p}^{\prime}$.  
Note that FindInput must output values of all non-pebbled nodes by Proposition~\ref{prop:thrifty-all-nodes-queried}.
Let $C$ and let $C^{\prime}$ be the accepting computation paths for $I$ and $I^{\prime}$ found by FindInput that 
passes through $s$.  Let $C_1$, $C_{1}^{\prime}$ and $C_2$, $C_{2}^{\prime}$ be the segments of $C$ and $C^{\prime}$ 
before and after $s$ respectively.

Suppose that $I$ and $I^{\prime}$ differ in the value of a black (resp. white) pebbled node $2i$ and $x$ and $x^{\prime}$ are 
the values of node $2i$ in $I$ and $I^{\prime}$ respectievely.  
Then the computation paths $c$ and $C^{\prime}$ are as shown in Figure~\ref{fig:black-diff} 
\begin{figure}[ht]
\includegraphics[scale=0.8]{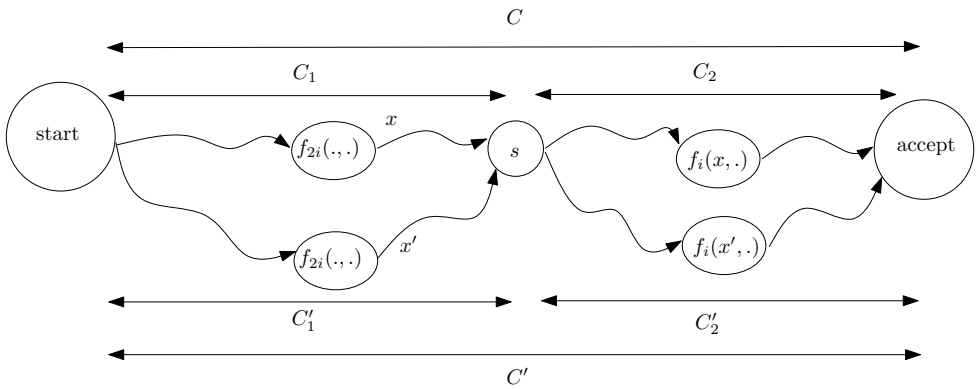}
\caption{Computation paths when a black pebbled node differs}
\label{fig:black-diff}
\end{figure}

\begin{figure}[ht]
\includegraphics[scale=0.8]{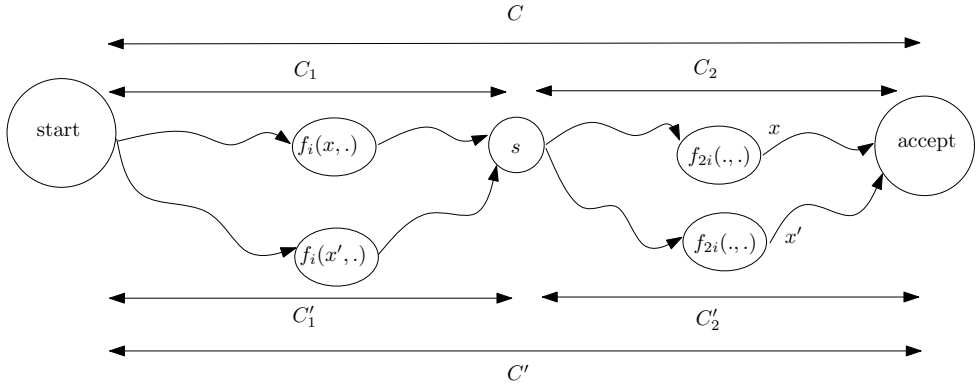}
\caption{Computation paths when a white pebbled node differs}
\label{fig:white-diff}
\end{figure}
(resp. Figure~\ref{fig:white-diff}) by Proposition~\ref{prop:thrifty-all-nodes-queried}.  
Now since $B$ is a RONTBP, the nodes queried in $C_1$ and $C_{2}^{\prime}$ are disjoint
and therefore we can construct an input $J$ with the accepting computation path $C_1 C_{2}^{\prime}$.  
But this path makes a non-thrifty query.
\end{proof}

\begin{theorem}
\label{thm:rontbp-lower-bound}
Any RONTBP solving \bt{h}{2}{k} must have at least $k^{\lceil
h/2 \rceil}$ states.  \end{theorem}

\begin{proof} 
We give a proof using the entropy method.  Let $B$ be a RONTBP solving \bt{h}{2}{k}.  
Our input set is the set $E$ given in 
Definition~\ref{def:hard-inputs}.  Now for each input $I \in E$, we choose an 
arbitrary accepting computation path $C(I)$ and map $I$ to the state $s$ in $C(I)$ such 
that the whole black-white pebbling configuration associated with $s$ has at least 
$\lceil h/2 \rceil$ pebbles on non-root nodes.  Such a state exists by the whole black-white 
pebbling lower bounds given by \cite{wehr-tep} \cite{frac-pbl-lb}.  Now we can conclude by 
Lemma~\ref{lem:state-input} that there are at most $k^{N - \lceil h/2 \rceil}$ inputs in $E$ reaching 
$s$ on an accepting computation path in the RONTBP $B$.  Therefore, there are at least 
$k^{\lceil h/2 \rceil}$ states in $B$.
\end{proof}

From Theorem~\ref{thm:rontbp-upper-bound} and Theorem~\ref{thm:rontbp-lower-bound}, we obtain the tight bound $\stheta{k^{\lceil h/2 \rceil}}$ for 
RONTBPs solving \bt{h}{2}{k}.

\section{Tight Bounds for BINTBP}
\label{sec:bintbp}

We prove upper bounds for BINTBP by showing that BINTBPs can implement 
fractional black-white pebbling of \T{h}{2} to solve \bt{h}{2}{k}

\begin{theorem} \bt{h}{2}{k} can be solved by a BINTBP using $\sbigoh{k^{h/2}}$
states.  \label{thm:bintbp-upper-bound} \end{theorem}
\begin{proof}
Cook et. al. \cite{wehr-tep} describes an NTBP that solves \bt{h}{2}{k} 
using $\sbigoh{k^{h/2}}$ states.  The idea is the same as the idea used in the proof 
of Theorem~\ref{thm:rontbp-upper-bound}.  
We claim that this NTBP is a BINTBP.  
To see this, consider any state $s$ in that BP and node $i$ in the tree 
that has a black pebble value of $b_i$ and a white pebble value of $w_i$.  Now
if we consider the set of inputs $F_s$ reaching $s$.  Then the fraction of bits
corresponding to $b_i$ can take only one particular value in $F_s$.  The rest
of the bits can take all possible combinations.  If we consider $A_s$, then the
fraction of bits corresponding to $b_i$ and $w_i$ are fixed (The value of
fraction of bits corresponding to $w_i$ will be the guessed value) and the rest
of the bits can take all possible combinations.
\end{proof}

Similar to Definition~\ref{def:rontbp-pbl-def}, we now define the fractional black-white pebbling configuration 
at a state in the BINTBP.

\begin{definition}[Fractional Black-White Pebbling Configuration at a State]
Let $B$ be a BINTBP solving \bt{h}{2}{k} and let $s$ be a state such that for some input $I \in E$ 
$s$ has atleast one accepting computation path for $I$ passing through $s$.  
Then for any non-root node $i$, we define the black and 
white pebble values for the configuration at state $s$ as follows.
\begin{align*} \label{eq:actual-pebble} b(i, s) &= 1 - \log_{k} |\proj{F_s}{i}|
\notag \\w(i, s) &= \log_{k} \frac{|\proj{F_s}{i}|}{|\proj{A_s}{i}|}
\end{align*}
\label{def:bintbp-pbl-def}
\end{definition}

Notice that in a minimal BINTBP, any state $s$ must have atleast one accepting 
computation path passing through it.  Otherwise, the state $s$ can be removed.  
Note that the pebbling configuration only depends on state $s$ by definition. 

We now claim that Definition~\ref{def:bintbp-pbl-def} of pebble values satisfy the restrictions
imposed on pebble values by~\eqref{eq:pebble-restr}.

\begin{claim} \label{claim:pbl-rng} For any non-root node $i$ and state $s$, $0
\leq b(i, s), w(i, s) \leq 1$.  \end{claim}

\begin{claim} \label{claim:pbl-sum} For any non-root node $i$ and state $s$,
$b(i, s) + w(i, s) \leq 1$.  \end{claim}

The following claim establishes the fact that if the total pebble value of the
tree (in non-root nodes) is high at a state, then there are only a few inputs
on an accepting computation path reaching that state.  In other words, if the
pebble value at a point of the computation is high, then the entropy at that
point is low.

\begin{claim} \label{claim:pbl-inputs} If the total pebble value of the
non-root nodes of the tree at a state $s$ is $p$, then the number of inputs $I \in E$ 
reaching $s$ on an accepting computation path is $k^{N - p}$.
\end{claim}
\begin{proof} Consider a particular non-root node $i$. Assume that
the total pebble value at $i$ is $p_i$.  From this we have $1 - \log_{k}
|\proj{F_s}{i}| + \log_{k} \frac{|\proj{F_s}{i}|}{|\proj{A_s}{i}|} = p_i$.
Therefore $|\proj{A_s}{i}| = k^{1 - p_i}$. Now by simple counting the total
number of inputs on an accepting computation path is $k^{\sum_{i = 2}^{N+1} (1
- p_i)} = k^{N - p}$.  \end{proof}

We now identify the fractional black-white pebbling of \T{h}{2} on an accepting 
computation path $C$ for an input $I \in E$.  First, we identify certain 
\emph{critical} states in the path $C$.  The pebbling will satisfy the criteria 
that the pebbling configuration changes (i.e., pebbling moves happen) only at 
critical states.  Then we will show that these pebbling configurations always 
underestimate the pebble values of nodes given by Definition~\ref{def:bintbp-pbl-def}.

\begin{definition}[Critical States for Nodes]
The critical state for the root node is the last state querying the root node.  
Every non-root node $j$ may have multiple critical states. 
 Let $s$ denote \emph{a} critical state of parent of
$j$.  If $b(j, s) > 0$, then the last node querying node $j$ before $s$ is a
critical state for $j$. If $w(j, s) > 0$, then the first node querying node 
$j$ after $s$ is a critical state for $j$.
\label{def:bintbp-critical-states}
\end{definition}

We will follow the convention that the start state and accepting state are critical.

For the lower bound proof we will work with the following fractional black-white pebbling along an 
arbitrary accepting computation path for an input in $E$.  
This pebbling satisfies the condition that pebble values are always underestimated. 

\paragraph*{Fractional Black-White Pebbling along Critical States}
We now define the pebbling along critical states 
on an accepting computation path of input $I$. The black pebble value of the root node becomes $1$ immediately 
after its critical state and it is immediately unpebbled.  
Now we define the pebble values of an arbitrary non-root node
$j$. Let $s^{\prime}$ be a critical state for $j^{\prime}$, the parent of $j$.
If $b = b(j, s^{\prime}) > 0$ (We say that $s^{\prime}$ needs this black pebble at $j$), 
then this black pebble value must have increased
from 0 to $b$ at some point of computation. Now consider the critical state $s$
for $j$ before $s^{\prime}$ as per Definition~\ref{def:bintbp-critical-states}.
The black pebble value of node
$j$ is increased from 0 to $b$ at the critical state immediately following $s$.  This state $s$ must 
exist as otherwise we have $b = 0$.  
Similarly, if $w = w(j, s^{\prime}) > 0$ (We say that $s^{\prime}$ needs this white pebble at $j$), then this white pebble value must
decrease from $w$ to 0 at some point of computation. Now consider the critical
state $s$ for $j$ after $s^{\prime}$ as per Definition~\ref{def:bintbp-critical-states}. The white pebble value
is reduced from $w$ to 0 at the critical state immediately following $s$.  This state $s$ must exist as 
otherwise we can construct an input using bitwise independence that differs from $I$ only in the value of node $j$ 
that has an accepting computation path 
with a non-thrifty query.  We decrease the black pebble values of all nodes 
if they are not needed further along the computation path and increase the white pebble values only at the critical state that 
needs them.  

The following claims about the validity of the starting and ending pebbling
configurations are easily proved.

\begin{claim} \label{claim:pbl-start} The start state has an empty pebbling
configuration.  \end{claim}

\begin{claim} \label{claim:pbl-acc} The accepting state has an empty pebbling
configuration.  \end{claim}

The following lemmas establish the fact that the pebbling sequence along critical states is a valid pebbling sequence.

\begin{lemma} \label{lemma:pbl-children} Let $s$ be a critical state for node
$j$, then both of $j$'s children are fully pebbled at $s$.  \end{lemma}
\begin{proof} Let $s$ query $f_j(u, v)$. We have $\proj{A_s}{2j} = \{ u \}$
(and $\proj{A_s}{2j+1} = \{ v \}$) by the thrifty property.  Then $b(2j, s) +
w(2j, s) = 1 - \log_{k} |\proj{F_s}{2j}| + \log_{k}
\frac{|\proj{F_s}{2j}|}{|\proj{A_s}{2j}|} = 1$ (and similarly for $2j + 1$).
\end{proof}

\begin{lemma} \label{lemma:pbl-incdec} If the black pebble value of node $j$ is
increased or the white pebble value of node $j$ is decreased at state $s$, then
both its children are fully pebbled at the critical state immediately before
$s$.  \end{lemma} \begin{proof} For a node $j$, the black pebble value is
increased or the white pebble value is decreased only at the critical state
immediately following a critical state for $j$. By
Lemma~\ref{lemma:pbl-children} both children of node $j$ are fully pebbled at
this critical state.  \end{proof}

The following is our key technical lemma and establishes the fact that the
pebbling values defined for the critical states never overestimate the actual
pebbling values of nodes.

\begin{lemma} \label{lemma:pbl-underest} Let $b$ and $w$ be the pebble values
 at state $s$ for an arbitrary non-root node $2j$ with respect to an arbitrary 
accepting computation path for some input in $E$, then $b \leq b(2j,
s)$ and $w \leq w(2j, s)$.  \end{lemma} \begin{proof} The proof is divided into
two parts. First, we show that the black pebble values are never overestimated.
Then we show that white pebble values are never overestimated.

We consider an arbitrary state $s$ at which the black pebble value of node $2j$
is defined as $b > 0$.  Note that the black pebble value of a non-root node
$2j$ is non-zero if and only if there exists a critical state for the parent of
$2j$ at which the actual pebble value of $2j$ is $b$. Therefore, there exists a
state $s_{2j}$ that is a critical state for $2j$ before $s$ and $s_{j}$ that is
a critical state for $j$, the parent of $2j$, after $s$ (with $s = s_{j}$
possibly.).  Now suppose that the actual black pebble value for node $2j$ at
state $s$ is $b(2j, s)$ and that $b(2j, s) < b$. 

\begin{eqnarray*} 1 - \log_{k} |\proj{F_s}{2j}| &<& b \\ \implies
|\proj{F_s}{2j}| &>& {k}^{1 - b} \end{eqnarray*}

Now by the independence assumption we may conclude that there are more than
${k}^{1 - b}$ inputs that differ only at the value of node $2j$ reaching $s$. By the
definition of critical states, there does not exist any node querying $2j$ in
$C(I)$ from $s$ to $s_{j}$.  All these inputs can follow the same path to the
critical state $s_{j}$. Therefore, the black pebble value is $b(2j, s_{j}) <
b$, a contradiction.

It remains to prove that white pebble values are never overestimated.  We will
prove that the white pebble value of a node $2j$ is at least the estimated
value $w$ between all states from $s_j$ to $s_{2j}$ (both inclusive). Here
$s_j$ is a critical state for $j$ at which node $2j$ acquired a white pebble
value of $w$ and $s_{2j}$ is the critical state for $2j$ after which this
pebble value is removed. In order to prove this, it is sufficient to prove that
the ratio $\frac{f^\prime}{a^{\prime}} =
\frac{|\proj{F_{s^\prime}}{2j}|}{|\proj{A_{s^\prime}}{2j}|}$ for any state
$s^\prime$ is greater than the corresponding ratio $\frac{f}{a}$ at state
$s_j$, where $s^\prime$ is a state on $C(I)$ in the segment from $s_j$ to
$s_{2j}$.  By the independence argument, we have $f^{\prime} \geq f$ by taking
projections of all $f$ inputs that differ from $I$ only at node $2j$. We will
show that if $a^{\prime} > a$, then $f^{\prime} > f$ by an appropriate amount
so that the ratio is not reduced.

Since the white pebble value is acquired at state $s_j$, we have
$\wpbl{s_j}{2j} = w$. Now consider a state $s^{\prime}$ (Possibly equal to
$s_{2j}$) on the segment of the computation path $C(I)$ between $s_j$ and
$s_{2j}$. Our aim is to prove that $w \leq \wpbl{s^{\prime}}{2j}$. Let $f =
|\proj{F_{s_j}}{2j}|$, $f^{\prime} = |\proj{F_{s^{\prime}}}{2j}|$, $a =
|\proj{A_{s_j}}{2j}|$ and $a^{\prime} = |\proj{A_{s^{\prime}}}{2j}|$. First of
all note that $f^{\prime} \geq f$ since there are no nodes querying $2j$ from
$s_j$ to $s_{2j}$ and the independence property guarantees $f$ inputs that
differ from $I$ only at node $2j$ will reach $s^{\prime}$.  Now we will show
that whenever $a^{\prime} > a$ , $f^{\prime}$ is greater than $f$ by the same
multiplicative factor. Note that both $f$ and $a$ are powers of two. By the
assumption of bitwise independence, we can partition bits of node $2j$ into
``fixed'' bits and ``unfixed'' bits for any $F_s$ (and $A_s$). The only way to
add elements to these sets are by unfixing bits. Let us assume that exactly one
more bit became unfixed in $\proj{A_{s^{\prime}}}{2j}$. So $a^{\prime} = 2a$.  

Let $r^{\prime}$ be a value in
$\proj{A_{s^{\prime}}}{2j}\setminus\proj{A_{s_j}}{2j}$. We claim that
$r^{\prime} \notin \proj{F_{s_j}}{2j}$. We will prove this by contradiction.
Suppose $r^{\prime} \in \proj{F_{s_j}}{2j}$, then by the independence property
there is an input $J$ which is the same as $I$ except that $\V{2j}{J} =
r^{\prime}$ reaches $s^{\prime}$ through $s_j$. Since $r^{\prime} \in
\proj{A_{s^{\prime}}}{2j}$, there is an accepting path for $J$ through $s_j$.
This accepting path is obtained by using the independence property of
$A_{s^{\prime}}$ and the fact that an accepting computation for $I$ passes
through $s^{\prime}$.  But this path makes a non-thrifty query at $s_j$.
Therefore $r^{\prime} \notin \proj{F_{s_j}}{2j}$ as claimed.  Since $r^{\prime}
\in \proj{F_{s^{\prime}}}{2j}$, at least one bit must have become unfixed. But
this implies $f^{\prime} \geq 2f$. This proof can be easily extended to the
case where $a^{\prime} = 2^{m}a$ for any $m$.  \end{proof}

Now we prove tight size lower bounds for BINTBPs solving \bt{h}{2}{k}.

\begin{theorem} If $B$ is a BINTBP solving \bt{h}{2}{k}, then $B$ has at least $\frac{1}{2}k^{h/2}$ states.
\label{thm:bintbp-lower-bound}
\end{theorem} \begin{proof} Assume that $k$ is a power of two. We now apply the
entropy method described in Subsection~\ref{subsec:entropy}.  Our input set is
the set $E$ described previously. We now describe our distribution function
$f$. Recall that each instance $I$ in $E$ is a ``yes'' instance and therefore
guaranteed to have an accepting computation path $C(I)$ in $B$. As we have
already seen, we can identify a sequence of critical states in $C(I)$ and
associate a fractional black-white pebbling configuration with each critical
state such that the sequence of fractional black-white pebbling configurations
form a valid fractional black-white pebbling of $\T{h}{2}$ (See Claims~\ref{claim:pbl-rng},
\ref{claim:pbl-sum}, \ref{claim:pbl-start}, \ref{claim:pbl-acc},
\ref{lemma:pbl-children}, and \ref{lemma:pbl-incdec}).  But we know that any
valid fractional black-white pebbling of $\T{h}{2}$ must have a configuration
with at least $h/2$ pebbles on non-root nodes \cite{frac-pbl-lb}.  Let $s$ be
the critical state in $C(I)$ that corresponds to this configuration. Our
distribution function $f$ maps $I$ to $s$. Now consider an arbitrary state $s$
in $range(f)$ and consider the set $G_s = f^{-1}(s)$. By
Claim~\ref{claim:pbl-inputs}, we have $|G_s| \leq k^{N-h/2}$.  It follows that
$B$ has at least $k^{h/2}$ states.

When $k$ is not a power of two, we consider the highest power of two
($2^{\ell}$) smaller than $k$. Consider the sub-BP of $B$ that solves
\bt{h}{2}{k} when the values are from the set $[2^{\ell}]$.  By Definition~\ref{def:bintbp} 
and the lower bound when $k$ is a power of two, we have
that this sub-BP of $B$ has at least $2^{{\ell}^{h/2}} > \frac{1}{2}k^{h/2}$
states.  \end{proof}

From Theorem~\ref{thm:bintbp-upper-bound} and Theorem~\ref{thm:bintbp-lower-bound}, we conclude that 
BINTBPs solving \bt{h}{2}{k} has size \stheta{k^{h/2}}.

\begin{remark} We note that the lower-bound proof in \cite{wehr-tep} for
deterministic thrifty BPs can be obtained by specializing our argument to
deterministic thrifty BPs. Specifically, we define the black pebble value of a
node as $1$ if and only if its value is known. The critical state for the root
node is the last state querying root and critical state for other nodes $j$ are
those states which query $j$ and immediately precedes the critical state for
$j$'s parent. The fact that the computation follows a valid black pebbling can
be argued using thriftiness (bitwise independence is not required.). We then
map each input to the state that has $h$ or more pebbles. The lower bound
follows.  \end{remark}

\section{Lower Bounds for Deterministic BPs Using Entropy Method}
\label{sec:otherproofs}

In this section, we show that many lower bound proofs in \cite{wehr-tep} can be
derived using the entropy method and derive some new applications of the
method.

\begin{theorem}{(\cite{wehr-tep})} Any deterministic $k$-way BP solving
\hft{3}{2}{k} has at least $k^3$ states. \label{thm:f3} \end{theorem}
\begin{proof}
First, we will consider a $k$-way BP that takes two inputs $u, v \in [k]$ and
computes $u +_k v$ where $+_k$ is addition modulo $k$. We will prove a size lower-bound of $k$ states for this problem. Then
we will use this result to prove the theorem.

Let $B$ be a $k$-way BP solving the above problem. We apply the entropy method
to prove the required size lower-bound.  Our input set $A$ consists of $k^2$
inputs (all inputs). Our distribution function maps each input in $A$ to the
last edge in the $k$-way BP $B$ solving this problem. Now consider an arbitrary
edge $e$ labelled $r$ and connecting a state labelled (w.l.g.) $u$ to the output
state $s$. Now consider the set of inputs $F_e$ reaching this edge. The only
possible inputs are those with $u = r$ and $u +_k v = s$.  But this implies
that $v = s -_k r$. Therefore $F_e = \{ (r, s -_k r) \}$ has cardinality one.
Since the choice of $e$ was arbitrary, we have $Edges(B) \geq k^2/1 = k^2$.
Since we are considering $k$-way BPs where each state has exactly $k$ outgoing
edges $States(B) \geq k$. 

Consider the sub-problem of \hft{3}{2}{k} where $f_1 = +_k$, leaves are allowed
to take arbitrary values, and for any input $I$, we allow $f_{j}^{I}(\V{2j}{I},
\V{2j+1}{I})$ for $j = 2, 3$ to take arbitrary values and restrict it to $1$
elsewhere. Consider a $k$-way BP $B$ solving this problem. Now consider the
sub-BP $b^{\prime}$ obtained from $B$ by fixing $(v_4, v_5) = (v_6, v_7) = (r,
s)$ for some $r, s \in [k]$. Note that the sub-BP $B^{\prime}$ computes $u +_k
v$ for $u = f_2(r, s)$ and $v = f_3(r, s)$ and therefore must have at least $k$
states. Now the set of all states querying $f_2$ or $f_3$ in $B$ is the
disjoint union of all states querying $f_2(r, s)$ and $f_3(r, s)$ for $k^2$
$(r, s)$ pairs. Therefore $States(B) \geq k^3$ as claimed.  \end{proof}

The \children{4}{2}{k} problem is the same as \hft{4}{2}{k} problem except that
the tree has no root node and the values at nodes $2$ and $3$ together is
defined as the output.

\begin{theorem}{(\cite{wehr-tep})} Any deterministic $k$-way BP solving
\children{4}{2}{k} has at least $k^4$ states. \label{thm:c4} \end{theorem}
\begin{proof} Consider a $k$-way BP that takes four inputs $u, v, w, x$ and
computes the tuple $(u +_k v, w +_k x)$.  We will prove a size lower-bound of
$k^2$ states for this problem and argue that the theorem follows from this
result.

Let $B$ be a deterministic $k$-way BP solving this problem. We now apply the
entropy method.  Our input set $A$ is the set of all inputs and therefore $|A|
= k^4$. Our distribution function will map each input in $A$ to the last edge
in its computation path on $B$. Consider an arbitrary edge $e$ labelled $r$
that connects a query state labelled $u$ to the output state $(s, t)$. Now
consider the set of inputs $F_e$ that get mapped to $e$. We have $u = r$, $v =
s -_k r$, and $w +_k x = t$. Since there are exactly $k$ inputs that satisfy
these conditions $|F_e| \leq k$. Therefore $Edges(B) \geq k^4/k = k^3$ and it
follows that $States(B) \geq k^2$. 

Consider the sub-problem of \children{4}{2}{k} where $f_2 = f_3 = +_k$, leaves
are allowed to take arbitrary values, and for any input $I$, we allow
$f_{j}^{I}(\V{2j}{I}, \V{2j+1}{I})$ for $j = 4, 5, 6, 7$ to take arbitrary
values and restrict it to $1$ elsewhere. Consider a $k$-way BP $B$ solving this
problem.  Now consider the sub-BP $B^{\prime}$ obtained from $B$ by fixing
values of sibling leaves to $(r, s)$.  Note that the sub-BP $B^{\prime}$ solves
the problem discussed in the previous paragraph and hence requires $k^2$
states. As before, since the level 2 query states of $B$ are the disjoint union
of query states for $k^2$ distinct $(r, s)$ pairs, we have $States(B) \geq
k^4$.
\end{proof}

We now present a new lower-bound of $\Omega(2^hk)$ for \sft{h}{2}{k} problem
when the function at internal nodes are restricted to a group operation.  Cook 
et. al. \cite{wehr-tep} has shown a lower bound of $\Omega(2^h)$ for this problem.

\begin{theorem} \label{thm:sft} Any deterministic $k$ way BP solving
\sft{h}{2}{k} with the functions at internal nodes restricted to a group
operation has at least  $2^{h-2} k$ states.  \end{theorem} \begin{proof} Assume
without loss of generality that the functions at internal nodes are $+_k$. The
leaf nodes are labelled $x_1 = 2^{h-1}, \ldots , x_{2^{h-1}} = 2^{h} - 1$. Let
$B$ be a deterministic $k$-way BP solving this problem.  Now consider the
sub-BP $B^{\prime}$ obtained from $B$ by fixing $x_3, \ldots , x_{2^{h-1}}$ to
$1$. The sub-BP $B^{\prime}$ computes $x_1 +_k x_2$ and therefore has at least
$k$ states. A similar argument can be applied to each pair of leaves. Since
there are $2^{h-2}$ disjoint pairs of leaves, the BP $B$ must have at least
$2^{h-2} k$ states.  \end{proof}

\noindent {\bf Upper Bounds:} We observe upper bounds for the size of branching programs computing
\sft{h}{2}{k} problem when the function at internal nodes are restricted to a
group operation. The associativity of the group operation implies upper bounds.
 We now briefly describe a BP for this problem. The BP is a layered BP of width $k$. The BP evaluates 
the group product in a left-associative fashion. In order to do this, the BP only has 
to remember the value of the product $v_1 \ldots v_{i-1}$ in the $i^{\text{th}}$ layer. 
This value is in $[k]$ and can be remembered using width $k$. Then, in the $i^{\text{th}}$ layer, the BP 
reads $v_i$ and moves to $i+1^{\text{st}}$ layer updating the remembered value as required.
There are two variations possible in this setting. In the first one, the
function at the internal nodes is fixed. In this case the branching
program described will be of size $2^hk$ and hence Theorem~\ref{thm:sft} is tight. In the
second version, when the function at the internal node is also a part of the
input, the described method will give an upper bound of $2^hk^2$ (since we also
have to query the function values).

\section{Conclusion}
We studied read-once nondeterministic thrifty branching programs solving \bt{h}{2}{k} and showed that this model captures exactly 
algorithms implementing a whole black-white pebbling strategy.  
We studied nondeterministic thrifty branching programs solving \bt{h}{2}{k} and showed that this model along with the bitwise independence 
restriction captures exactly algorithms implementing a fractional black-white pebbling strategy.  
These results extend the result in 
\cite{wehr-tep} that deterministic thrifty branching programs solving \bt{h}{2}{k} captures exactly algorithms implementing a 
black pebbling strategy to solve this problem. 

Our work is also the first instance where the entropy method, introduced by 
Jukna and Z{\'a}k, is applied to obtain size lower bounds for a nondeterministic computation model.  We also give a simplified 
and unified view of many of the existing size lower bound proofs for branching programs solving the tree evaluation problem.

One of the main open problems that arises out of our work is to understand how restricted is the {\em bitwise independence restriction} on nondeterministic branching programs solving the tree evaluation problem.
Following the thrifty hypothesis in the deterministic world, it is possible that the best nondeterministic branching programs are thrifty and hence it might suffice to prove lower bounds against thrifty versions of the branching program. Although we found that all known nondeterministic branching programs solving tree evalutation problem are bitwise independent, it is conceivable that there is a smaller nondeterministic thrifty branching program without the bitwise independence restriction.

\bibliography{references}
\end{document}